\documentclass[final,twoside,11pt]{entics}
\usepackage{enticsmacro}
\usepackage{graphicx}
\usepackage[all]{xy}
\usepackage{amsmath}
\usepackage{mathtools}
\usepackage{relsize}
\usepackage{tikz-cd}
\usepackage{cite}
\usepackage{inconsolata}
\usepackage{scalefnt}

\newcommand{\UU}{\mathcal{U}}
\newcommand{\VV}{\mathcal{V}}
\newcommand{\WW}{\mathcal{W}}
\newcommand{\TT}{\mathcal{T}}
\newcommand{\KK}{\mathcal{K}}

\newcommand{\blank}{{-}}

\newcommand{\setof}[1]{\left\{ #1 \right\}}

\newcommand{\IsPropNm}{\mathsf{isProp}}
\newcommand{\IsProp}[1]{\IsPropNm\left(#1\right)}

\newcommand{\ttmodule}[1]{\href{https://www.cs.bham.ac.uk/~mhe/TypeTopology/#1.html}{\texttt{#1}}}

\newcommand{\is}{\coloneqq}

\newcommand{\comment}[1]{}

\newcommand{\trunc}[1]{\left\| #1 \right\|}
\newcommand{\abs}[1]{\left| #1 \right|}

\newcommand{\closednucl}[1]{`#1\textrm'}
\newcommand{\opennucl}[1]{\neg`#1\textrm'}

\newcommand{\EmptyTy}{\mathbf{0}}
\newcommand{\UnitTy}{\mathbf{1}}
\newcommand{\NatTy}{\mathbb{N}}
\newcommand{\ListTy}[1]{\mathsf{List}(#1)}

\newcommand{\ddnm}{\mathfrak{o}}
\newcommand{\dd}[1]{\ddnm(#1)}

\newcommand{\Spec}{\mathbf{Spec}}
\newcommand{\Stone}{\mathbf{Stone}}
\newcommand{\Frm}{\mathbf{Frm}}
\newcommand{\Loc}{\mathbf{Loc}}

\newcommand{\McK}{\mathcal{K}}

\newcommand{\emptyl}{\varepsilon}
\DeclareMathOperator{\cons}{\mathbf{\colon\kern-1.0ex\colon\kern-0.3ex}}
\DeclareMathOperator{\append}{}

\newcommand{\define}[1]{\emph{#1}}

\newcommand{\paren}[1]{\left( #1 \right)}

\newcommand{\opens}[1]{\mathcal{O}(#1)}

\usepackage{scalerel}[2016/12/29]

\newcommand\bigexistsp{%
  \mathop{\lower1.2ex\hbox{\ensuremath{%
        \mathlarger{\mathlarger{\mathlarger{\scaleobj{1.2}{\exists}}}}}}}%
  \limits}

\newcommand\ScaleExists[1]{\vcenter{\hbox{\scalefont{#1}$\exists$}}}

\DeclareMathOperator*\bigexists{%
  \vphantom\sum
  \mathchoice{\ScaleExists{2}}{\ScaleExists{1.4}}{\ScaleExists{1}}{\ScaleExists{0.75}}}

\newcommand{\FamNm}{\mathsf{Fam}}
\newcommand{\Fam}[2]{\FamNm_{#1}\left(#2\right)}

\newcommand{\Patch}{\mathsf{Patch}}
\newcommand{\opposite}[1]{{#1}^{\mathsf{op}}}

\newcommand{\WayBelow}{\ll}
\newcommand{\WellInside}{\eqslantless}

\def\conf{MFPS 2022} 	
\volume{1}			
\def\volu{1}			
\def\papno{16}			
\def\mydoi{{\normalshape  \href{https://doi.org/10.46298/entics.10808}{10.46298/entics.10808}}}
\def\copynames{A. Tosun, M. Escard\'o} 
\def\runtitle{Patch Locale of a Spectral Locale in Univalent Type Theory}

\def\lastname{Tosun \and Escard\'o}
\begin{document}
\begin{frontmatter}
  \title{Patch Locale of a Spectral Locale\\ in\\ Univalent Type Theory}
  \author{Ayberk Tosun\thanksref{a}\thanksref{myemail}}	
  \thanks[myemail]{%
    Email:
     \href{mailto:a.tosun@pgr.bham.ac.uk}{\texttt{\normalshape a.tosun@pgr.bham.ac.uk}}
  }
  \author{Mart\'in H.\ Escard\'o\thanksref{a}\thanksref{coemail}}
  \address[a]{School of Computer Science\\ University of Birmingham\\ Birmingham, United Kingdom}
  \thanks[coemail]{Email:  \href{mailto:m.escardo@cs.bham.ac.uk} {\texttt{\normalshape
        m.escardo@cs.bham.ac.uk}}}
\begin{abstract}
  Stone locales together with continuous maps form a coreflective subcategory of
  spectral locales and perfect maps. A proof in the internal language of an
  elementary topos was previously given by the second-named author. This proof
  can be easily translated to univalent type theory using \emph{resizing
  axioms}. In this work, we show how to achieve such a translation
  \emph{without} resizing axioms, by working with large, locally small, and
  small complete frames with small bases. This turns out to be nontrivial and
  involves predicative reformulations of several fundamental concepts of locale
  theory.
\end{abstract}
\begin{keyword}
  locale theory, pointfree topology, patch locale, spectral locale, stone space,
  univalent type theory
\end{keyword}
\end{frontmatter}

\section{Introduction}\label{sec:intro}

The category $\Stone$ of Stone locales together with continuous maps forms a
coreflective subcategory of the category $\Spec$ of spectral locales and
\emph{perfect} maps i.e.\ maps preserving compact opens. A proof in the internal
language of an elementary topos was previously constructed in~\cite{patch-short,
  patch-full}, defining the patch frame as the frame of Scott continuous nuclei
on a given frame.

The objective of this paper is to carry out this construction in predicative,
constructive univalent foundations. In the presence of Voevodsky's resizing
axioms~\cite{voevodsky-resizing}, it is straightforward to translate the above
proof to univalent type theory. However, at the time of writing, there is no
known constructive interpretation of the resizing axioms. In such a predicative
situation, the usual approach to locale theory is to work with presentations of
locales, known as formal topologies~\cite{int-formal-spaces, coq-sambin,
  coquand-tosun}. However, we show that it is possible to work with locales
directly, if we adopt large, locally small, and small complete frames with small
bases~\cite{predicative-aspects}. This requires a number of substantial
modifications to the proofs and constructions of~\cite{patch-short, patch-full}:
\begin{enumerate}
  \item The patch is defined as the frame of Scott continuous nuclei
    in~\cite{patch-short, patch-full}. In order to prove that this is indeed a
    frame, one starts with the frame of all nuclei, and then exhibits the
    Scott continuous nuclei as a subframe. However, this procedure does not seem
    to be possible in our predicative setting as it is not clear whether all
    nuclei form a frame; so we construct the frame of Scott continuous nuclei
    \emph{directly}, which requires reformulations of all proofs about it
    inherited from the frame of all nuclei.
  \item In the impredicative setting, any frame has all Heyting implications,
    which is needed to construct open nuclei. Again, this does not seem to be
    the case in the predicative setting. We show, however, that it is possible
    to construct Heyting implications in locally small frames with small bases,
    by an application of the Adjoint Functor Theorem \label{item:aft} for posets.
 \item Similar to (\ref{item:aft}), we use the Adjoint Functor Theorem for posets to define
   the right adjoint of a frame homomorphism, using which we define the notion
   of a \emph{perfect map}, namely, a map whose defining frame homomorphism's
   right adjoint is Scott continuous. This notion is used in \cite{patch-short,
   patch-full}.
\end{enumerate}

For the purposes of this work, a \emph{spectral locale} is a locale in which the
compact opens form a \emph{small basis} closed under finite meets. A continuous
map of spectral locales is \emph{spectral} if its defining frame homomorphism
preserves compact opens. A \emph{Stone locale} is one that is compact and
zero-dimensional (i.e.\ whose clopens form a basis). Every Stone locale is
spectral since the clopens coincide with the compact opens in Stone locales. The
patch frame construction is the right adjoint to the inclusion $\Stone
\hookrightarrow \Spec$. The main contribution of our work is the construction of
this right adjoint in the predicative context of univalent type theory. We have
also formalised the development of this paper in the \textsc{Agda} proof
assistant~\cite{agda}, though our presentation here is self-contained and can be
followed independently of the formalisation. Although we have omitted some
proofs for lack of space, we have included all the crucial differences
from~\cite{patch-short, patch-full} in full.

The organisation of this paper is as follows. In Section~\ref{sec:foundations},
we present the type-theoretical context in which we work. In
Section~\ref{sec:spec-and-stone}, we present our definitions of spectral and
Stone locales that provide a suitable basis for a predicative development. In
Section~\ref{sec:aft}, we present a predicative version of the Adjoint Functor
Theorem for the simplified context of locales that is central to our
development. In Section~\ref{sec:meet-semilattice}, we define the
meet-semilattice of perfect nuclei as preparation for the complete lattice of
perfect nuclei, which we then construct in Section~\ref{sec:joins}. Finally in
Section~\ref{sec:coreflection}, we prove the desired universal property, namely,
that the patch locale exhibits the category $\Stone$ as a coreflective
subcategory of $\Spec$.

\section{Foundations}\label{sec:foundations}

In this section, we present the type-theoretical setting in which we work and
then provide the type-theoretical formulations of some of the preliminary
notions that form the basis of our work. Our type-theoretical conventions follow
those of de Jong and Escard\'o \cite{dejong-escardo-domains} and the Univalent
Foundations Programme \cite{hottbook}.

We work in Martin-L\"of Type Theory with binary sums $\blank + \blank$,
dependent products $\prod$, dependent sums $\sum$, the identity type $\blank = \blank$,
and inductive types including the empty type $\EmptyTy$, the unit type
$\UnitTy$, and the type $\ListTy{A}$ of lists over any type $A$. We
adhere to the convention of \cite{hottbook} of using $\blank \equiv \blank$ for
judgemental equality and $\blank = \blank$ for the identity type.

We work explicitly with universes, for which we adopt the convention of using
the variables $\UU, \VV, \WW$, and $\TT$. The ground universe is denoted $\UU_0$
and the successor of a given universe $\UU$ is denoted $\UU^+$. The least upper
bound of two universes is given by the operator $\blank \sqcup \blank$ which is
assumed to be associative, commutative, and idempotent. Furthermore, $(\blank)^+$
is assumed to distribute over $\blank \sqcup \blank$. Universes are computed for the
given type formers as follows:
\begin{itemize}
  \item Given types $X : \UU$ and $Y : \VV$, the type $X + Y$ inhabits universe
    $\UU \sqcup \VV$.
  \item Given a type $X : \UU$ and an $X$-indexed family, $Y : X \rightarrow \VV$, both
    $\sum_{x : X}Y(x)$ and $\prod_{x : X}Y(x)$ inhabit the universe $\UU \sqcup \VV$.
  \item Given a type $X : \UU$ and inhabitants $x, y : X$, the identity type
    $x = y$ inhabits universe $\UU$.
  \item The type $\NatTy$ of natural numbers inhabits $\UU_0$.
  \item The empty type $\EmptyTy$ and the unit type $\UnitTy$ have copies in
    every universe $\UU$, which we occasionally make explicit using the
    notations $\EmptyTy_{\UU}$ and $\UnitTy_{\UU}$.
  \item Given a type $A : \UU$, the type $\ListTy{A}$ inhabits $\UU$.
\end{itemize}

We assume only function extensionality, propositional extensionality and quotients, and do
not need full univalence for our development. We always maintain a distinction
between structure and property, and reserve logical connectives for
propositional types i.e.\ types $A$ satisfying $\IsProp{A} \is \prod_{x, y : A} x =
y$. We denote by $\Omega_\UU$ the type of propositional types in universe $\UU$
i.e.\ $\Omega_\UU \is \Sigma_{A : \UU} \IsProp{A}$.

We assume the existence of propositional truncation, given by a type former
$\trunc{\blank} : \UU \rightarrow \UU$ and a unit operation $\abs{\blank} : A \rightarrow
\trunc{A}$. The existential quantification operator is defined using
propositional truncation as:
\begin{equation*}
  \bigexists_{x : A} B(x) \quad\is\quad \trunc{\sum_{x : A}B(x)}.
\end{equation*}

When presenting proofs informally, we adopt the following conventions for
avoiding ambiguity between propositional and non-propositional types:
\begin{itemize}
  \item For the anonymous inhabitation $\abs{A}$ of a type, we say that $A$ is
    inhabited;
  \item For truncated $\Sigma$ types, we use the terminologies \emph{there is}
    and \emph{there exists}.
\end{itemize}

\subsection{Directed families}

We now proceed to define some preliminary notions in the type-theoretical
setting that we have just presented.

\begin{defn}[Family]\label{defn:family}
  A \define{$\UU$-family on a type} $A$ is a pair $(I, f)$ where $I : \UU$ and
  $f : I \rightarrow A$. We denote the type of $\UU$-families on type $A$ by
  $\FamNm_\UU(A)$ i.e.\ $\Fam{\UU}{A} \is \sum_{(I : \UU)} I \rightarrow A$.
\end{defn}

We often use the shorthand $\{ x_i \}_{i : I}$ for families. In other words,
instead of writing $(I, f)$ for a family, we write $\{ x_i \}_{i : I}$ where
$x_i$ denotes the application $f(i)$.

\begin{defn}[Subfamily]
  By a \define{subfamily} of some $\UU$-family $(I, f)$ we mean a family $(J, f
  \circ g)$ where $(J, g)$ is itself a $\UU$-family on $I$.
\end{defn}

When considering a subfamily $J$ of some family $\{ x_i \}_{i : I}$, we often
use the abbreviation $\{ x_j \mid j \in J \}$.

As mentioned in the introduction, Scott continuity plays a central role in our
development. To define Scott continuity, we define the notion of a directed
family. The definition that we work with (also used by de Jong and
Escard\'o~\cite{dejong-escardo-domains}) is the following:

\begin{defn}[Directed family]\label{defn:directed}
  Let $\{ x_i \}_{i : I}$ be a family in some type $A$ that is equipped with a
  preorder $\blank \le \blank$. The family $\{ x_i \}_{i : I}$ is called directed
  if (1) $I$ is inhabited, and (2) for every $i, j : I$, there exists some $k :
  I$ such that $x_k$ is the upper bound of $\{ x_i, x_j \}$.
\end{defn}

\subsection{Definition of locale}

A locale is a notion of space characterised solely by its frame of opens. Our
definition of a frame is parameterised by three universes: (1) for the carrier
set, (2) for the order, and (3) for the index types of families on which the join
operation is defined. We adopt the convention of using the universe variables
$\UU$, $\VV$, and $\WW$ for these respectively. We often omit universe levels in
contexts where they are not relevant to the discussion. In cases where only the
index universe $\WW$ is relevant, we speak of a $\WW$-locale for the sake of
brevity and omit universes $\UU$ and $\VV$.

\begin{defn}[Frame]
  A \define{$(\UU, \VV, \WW)$-frame} $L$ consists of:
  \begin{itemize}
    \item a set $| L | : \UU$,
    \item a partial order $\blank \le \blank : | L | \rightarrow | L | \rightarrow \Omega_\VV$,
    \item a top element $\top : | L |$,
    \item an operation $\blank \wedge \blank : | L | \rightarrow | L | \rightarrow | L |$ giving the
      greatest lower bound $U \wedge V$ of any two $U, V : | L |$,
    \item an operation $\bigvee\_ : \Fam{\WW}{| L |} \rightarrow | L |$ giving the least upper
      bound $\bigvee_{i : I} U_i$ of any $\WW$-family $\{ U_i \}_{i : I}$,
  \end{itemize}
  such that binary meets distribute over arbitrary joins, i.e.
  \begin{equation*}
    U \wedge \bigvee_{i : I} V_i = \bigvee_{i : I} U \wedge V_i
  \end{equation*}
  for every $U : | L |$ and $\WW$-family $\{ V_i \}_{i : I}$ in $| L |$.
\end{defn}

It follows automatically from the antisymmetry condition for partial orders that
the underlying type of a frame is a set. Finally, we note that most of our
results are restricted to $(\UU^+, \UU, \UU)$-frames for a fixed universe $\UU$,
which we refer to as \emph{large, locally small, and small complete} frames.
Even though some of our results apply to frames of a more general form, we
refrain from presenting the specific level of generality for the sake of
brevity. For the precise universe levels, we refer the reader to the
formalisation.

\begin{defn}[Frame homomorphism]
  Let $K$ and $L$ be a $(\UU, \VV, \WW)$-frame and a $(\UU', \VV', \WW)$-frame
  respectively. A function $h : |K| \rightarrow |L|$ is called a \define{frame
    homomorphism} if it preserves the top element, binary meets, and joins of
  $\WW$-families. We denote the category of frames and their homomorphisms by
  $\Frm$.
\end{defn}

We adopt the notational conventions of~\cite{sheaves}. A \define{locale} is a
frame considered in the opposite category called $\Loc \is \opposite{\Frm}$. To
highlight this, we adopt the standard convention of using the letters $X, Y, Z,
\ldots$ (or sometimes $A, B, C, \ldots$) for locales and denoting by $\opens{X}$ the frame
corresponding to a locale $X$. For variables that range over the frame of opens
of a locale $X$, we use the letters $U, V, W, \ldots$ We use the letters $f$ and $g$
for continuous maps $X \rightarrow Y$ of locales. A continuous map $f : X \rightarrow Y$ is given by
a frame homomorphism $f^* : \opens{Y} \rightarrow \opens{X}$.

\begin{defn}[Nucleus]\label{defn:nucleus}
  A \define{nucleus on a locale} $X$ is an endofunction $j : \opens{X} \to
  \opens{X}$ that is inflationary, idempotent, and preserves binary meets.
\end{defn}

In Section~\ref{sec:joins}, we will work with inflationary and
binary-meet-preserving functions that are not necessarily idempotent. Such
functions are called \emph{prenuclei}. We also note that, to show a prenucleus
$j$ to be idempotent, it suffices to show $j(j(U)) \le j(U)$ as the other
direction follows from inflationarity. In fact, the notion of a
nucleus could be defined as a prenucleus satisfying the inequality $j(j(U)) \le
j(U)$, but we define it as in Definition~\ref{defn:nucleus} for the sake of
simplicity and make implicit use of this fact in our proofs of idempotency.

\section{Spectral and Stone locales}\label{sec:spec-and-stone}

We start by defining the notion of a small basis for a frame. This is crucial
not just for the definitions of spectral and Stone locales that we use in our
development, but also for the Adjoint Functor Theorem that we present in
Section~\ref{sec:aft}.

\begin{defn}[Small basis]\label{defn:basis}
  Given a $\WW$-locale $X$, a $\WW$-family $\{ B_i \}_{i : I}$ of opens of $X$
  is said to \define{form a basis for $\opens{X}$} if
  \begin{equation*}
    \prod_{U : \opens{X}}
      \bigexists_{J : \Fam{\WW}{I}}
        U = \bigvee \{ B_j \mid j \in J \}.
  \end{equation*}
  A $\WW$-locale $X$ is then said to have a \define{small basis} if there exists
  a $\WW$-family $\{ B_i \}_{i : I}$ in $\opens{X}$ that forms a basis for
  $\opens{X}$.
\end{defn}

Given an open $U : \opens{X}$ with a small basis, we refer to the family $\{ B_j
\mid j \in J \}$ giving $U$ as its join as the \define{basic covering family for
$U$}.

It is important to note here that we use propositional truncation when defining
the notion of a locale having a basis. So even though we often speak of a
``locale with some small basis $\{ B_i \}_{i : I}$'', the existence of this
basis is a property meaning we have access to it only in contexts where the goal
is itself a proposition.

We often need covering families given by a basis to be directed. This is easy to
achieve if we work with bases closed under finite joins, which we can do without
loss of generality, as this closure produces another basis.

The standard impredicative definition of a spectral locale is as one in which
the compact opens form a basis closed under binary meets. To talk about
compactness, we define the \emph{way below} relation:

\begin{defn}[Way below]\label{defn:way-below}
  Given a $\WW$-locale $X$ and opens $U, V : \opens{X}$, $U$ is said to be
  \define{way~below} $V$, written $U \WayBelow V$, if
    \(\prod_{(I, f) : \Fam{\WW}{\opens{X}}}
      (I, f)\ \text{\textsf{directed}} \rightarrow V \le \bigvee (I, f) \rightarrow \bigexists_{i : I} U \le f(i)\).
\end{defn}

\begin{prop}
  Given any two opens $U$ and $V$ of a locale, the type $U \WayBelow V$ is a
  proposition.
\end{prop}

The statement $U \ll V$ is thought of as expressing that $U$ is compact
\emph{relative to} $V$. An open is said to be compact if it is compact relative
to itself:

\begin{defn}[Compact open of a locale]\label{defn:compact-open}
  An open $U : \opens{X}$ is called \define{compact} if $U \WayBelow U$.
\end{defn}

We denote the type of compact opens of a locale $X$ by $\McK(X)$. We adopt the
convention of using letters $C, D, \ldots$ for compact opens.

\begin{defn}[Compact locale]\label{defn:compact-space}
  A locale $X$ is called \define{compact} if its top element $\top : \opens{X}$ is
  compact.
\end{defn}

The standard definition of a spectral locale as one in which the compact opens
form a basis closed under finite meets is problematic in our predicative
setting, as it is not always the case that the type of compact opens of a $(\UU,
\VV, \WW)$-locale lives in $\WW$. In particular, the type of compact opens of a
$(\UU^+, \UU, \UU)$-locale lives in $\UU^+$ and it is accordingly said to be
\emph{large}. To address this problem, we restrict attention to locales with
small bases and express the notion of spectrality by imposing the conditions of
interest on the basic elements instead.

\begin{defn}[Spectral locale]\label{defn:spectral-locale}
  A locale $X$ is said to be \emph{spectral} if there exists a small basis $\{
  B_i \}_{i : I}$ such that:
  \begin{enumerate}
    \item every $B_i$ is compact, \label{item:spec-comp} and
    \item $\{ B_i \}_{i : I}$ is closed under finite meets i.e.\ there is $t :
      I$ with $B_{t} = \top$ and for any two $i, j : I$, there is $k : I$ such
      that $B_k = B_i \wedge B_j$. \label{item:spec-closed}
  \end{enumerate}
\end{defn}

We have previously remarked that we can assume without loss of generality
that bases of locales are closed under finite joins. Note here that this
assumption can also be made for bases of spectral locales as compact opens are
also closed under finite joins.

Spectral locales together with spectral maps constitute the category
$\Spec$. We now define the notion of a spectral map.

\begin{defn}[Spectral map]\label{defn:spectral-map}
  A continuous map $f : X \rightarrow Y$ between spectral locales $X$ and $Y$ is called
  \emph{spectral} if $f^*(V) : \opens{X}$ is a compact open of $X$ whenever $V$
  is a compact open of $Y$.
\end{defn}

A natural question to ask about our definition of spectral locale is whether it
corresponds to the previous informal definition: can there be compact opens that
\emph{do not} fall in the basis?

\begin{prop}\label{prop:cmp-bsc}
  For any spectral locale $X$, every compact open of $X$ falls in the basis.
\end{prop}
\begin{proof}
  Let $X$ be a spectral locale and denote by $\{ B_i \}_{i : I}$ its basis
  closed under finite joins. Let $C : \opens{X}$ be a compact open and let $\{
  B_j \}_{j \in J}$ be the covering family for $C$. Because the basis is closed
  under finite joins, this family is directed. As $C \le \bigvee_{i : I} B_i$ there must
  be some $k : I$ by the compactness of $C$ such that $C \le B_k$. It is also
  clearly the case that $B_k \le C$ and so $B_k = C$, meaning $C$ falls in the
  basis.
\end{proof}

\subsection{Zero-dimensional and regular locales}

Clopenness is central to the notion of a zero-dimensional locale, similar to the
fundamental role played by the notion of a compact open in the definition of a
spectral locale. To define the clopens, we first define the \emph{well inside}
relation.

\begin{defn}[Well inside relation]\label{defn:well-inside}
  Given a locale $X$ and opens $U, V : \opens{X}$, $U$ is said to be
  \define{well inside} $V$ (written $U \WellInside V$) if
  \begin{equation*}
    \bigexists_{W : \opens{X}} \left(U \wedge W = \bot\right)\ \times\ \left(V \vee W = \top\right).
  \end{equation*}
\end{defn}

\begin{defn}[Clopen]\label{defn:clopen}
  An open $U$ is called a \define{clopen} if it is well inside itself,
  which amounts to saying that it has a Boolean complement.
\end{defn}

Before we proceed to defining zero-dimensionality, we record the following
important fact about the well inside relation:

\begin{prop}\label{prop:well-inside-upwards-downwards}
  Given opens $U, V, W : \opens{X}$ of a locale $X$,
  \begin{enumerate}
    \item if $U \WellInside V$ and $V \le W$ then $U \WellInside W$; and
    \item if $U \le V$ and $V \WellInside W$ then $U \WellInside W$.
  \end{enumerate}
\end{prop}

Our definition of zero-dimensionality is analogous to the definition of a
spectral locale where conditions of interest apply only to basic opens.

\begin{defn}[Zero-dimensional frame]\label{defn:zero-dimensional}
  A locale is called \define{zero-dimensional} if it has a small
  basis $\{ B_i \}_{i : I}$ with each $B_i$ clopen.
\end{defn}

Zero-dimensionality can in fact be viewed as a special case of
\emph{regularity}. For purposes of our development, we need the result that $U
\WayBelow V$ implies $U \WellInside V$ in any zero-dimensional
locale~\cite[Lemma VII.3.5, pg.~303]{stone-spaces}. As this can be strengthened
to apply to the more general case of regular locales, we now define the notion
of regularity, using which we obtain a result slightly more general than needed.

\begin{defn}[Regular locale]\label{defn:regular}
  A locale is called \define{regular} if it has some basis $\{ B_i
  \}_{i : I}$ such that for any open $U$, every $B_j$ in the
  covering family for $U$ is well inside $U$.
\end{defn}

Similar to the case of spectral locales, the basis of a regular locale can be
assumed to be closed under finite joins without loss of generality as
every basis can be closed under finite joins to obtain another basis satisfying
the regularity condition of Definition~\ref{defn:regular}.

\begin{prop}\label{prop:zero-dimensional-implies-regular}
  Every zero-dimensional locale is regular.
\end{prop}
\begin{proof}
  Let $X$ be a zero-dimensional locale and call its basis $\{ B_i \}_{i : I}$.
  Consider some $U : \opens{X}$. There must be a basic covering $U = \bigvee_{i
    \in J} B_j$ such that each $B_j$ is clopen for every $j \in J$. Clearly,
  $B_j \le U$ so we have $B_j \WellInside B_j \le U$ which implies $B_j
  \WellInside U$ (by Proposition~\ref{prop:well-inside-upwards-downwards}.i).
\end{proof}

The following two propositions are needed to prove that compact opens and
clopens coincide in Stone locales, which we will need later. They are
adaptations of standard proofs \cite[pg.~303, Lemma VII.3.5]{stone-spaces} into
our predicative setting.

\begin{prop}\label{prop:way-below-implies-well-inside}
  In any regular locale, $U \WayBelow V$ implies $U \WellInside V$ for any
  two opens $U, V$.
\end{prop}
\begin{proof}
  Let $\{ B_i \}_{i : I}$ be the basis, closed under finite joins, of a regular
  locale $X$, let $U, V : \opens{X}$ such that $U \WayBelow V$, and let $\{
  B_j \}_{j \in J}$ be the basic family covering $V$. As $V \le \bigvee_{j \in
    J} B_j$ there must exist some $k \in J$ such that $U \le B_k$ by the fact
  that $U \WayBelow V$. We then have $U \le B_k \WellInside V$ which implies $U
  \WellInside V$ by Proposition~\ref{prop:well-inside-upwards-downwards}.
\end{proof}

\begin{prop}\label{prop:well-inside-implies-way-below}
  In any compact locale, $U \WellInside V$ implies $U \WayBelow V$ for any
  two opens $U, V$.
\end{prop}

The proof of Proposition~\ref{prop:well-inside-implies-way-below} is omitted as
it is exactly the same as in \cite[pg.~303]{stone-spaces}.

\begin{defn}[Stone locale]\label{defn:stone}
  A \define{Stone locale} is one that is compact and zero-dimensional.
\end{defn}

\begin{prop}
  In any Stone locale, an open is compact iff it is clopen.
\end{prop}
\begin{proof}
  By propositions \ref{prop:way-below-implies-well-inside} and
  \ref{prop:well-inside-implies-way-below} and the fact that every
  zero-dimensional locale is regular
  (Proposition~\ref{prop:zero-dimensional-implies-regular}).
\end{proof}

\section{Adjoint Functor Theorem for frames with small bases}\label{sec:aft}

We start with the definition of the notion of an adjunction in the simplified
context of posetal categories.

\begin{defn}
  Let $P$ and $Q$ be two posets. An adjunction between $P$ and $Q$ consists of a
  pair of monotonic maps $f : P \rightarrow Q$ and $g : Q \rightarrow P$
  satisfying $f \dashv g \is \prod_{x : P} \prod_{y : Q} f(x) \le y
  \leftrightarrow x \le g(y)$.
\end{defn}

In locale theory, it is standard convention to denote by $f_* : \opens{X} \rightarrow
\opens{Y}$ the right adjoint of a frame homomorphism $f^* : \opens{Y} \rightarrow
\opens{X}$ corresponding to a continuous map of locales $f : X \rightarrow Y$. The right
adjoint of a frame homomorphism is defined using the Adjoint Functor Theorem
which amounts to the definition: $f_* \is U \mapsto \bigvee \{ V : \opens{Y} \mid f^*(V) \le U
\}$. In the predicative setting of type theory however, it is not clear how the
right adjoint of a frame homomorphism would be defined as the family $\{ V :
\opens{Y} \mid f^*(V) \le U \}$ might be too big in general, meaning it is not clear
\emph{a priori} that its join in $\opens{X}$ exists. To resolve this problem, we
restrict attention once again to frames with small bases in which we circumvent
this problem by quantifying over only the basic elements.

\begin{thm}[AFT]\label{thm:aft}
  Let $X$ and $Y$ be two large, locally small, and small complete locales and
  let $f^* : \opens{Y} \to \opens{X}$ be a monotone map. Assume that $Y$ has a
  small basis $\{ B_i \}_{i : I}$. The map $f^*$ has a right adjoint iff
  $f^*(\bigvee_i U_i) = \bigvee_i f^*(U_i)$ for any small family $\{ U_i \}_{i :
  I}$ in $\opens{Y}$.
\end{thm}
\begin{proof}
  Let $f^* : \opens{Y} \rightarrow \opens{X}$ be a monotone map from frame
  $\opens{Y}$ to frame $\opens{Y}$ and assume that $Y$ has a small basis $\{ B_i
  \}_{i : I}$.

  The forward direction is easy: suppose $f^* : \opens{Y} \rightarrow \opens{X}$ has a
  right adjoint $f_* : \opens{X} \rightarrow \opens{Y}$. Let $\{ U_i \}_{i : I}$ be a
  family in $\opens{Y}$. By the uniqueness of joins, it is sufficient to show
  that $f^*(\bigvee_i U_i)$ is the join of $\{ f^*(U_i) \}_{i : I}$. It is clearly an
  upper bound by the fact that $f^*$ is monotone. Given any other upper bound
  $V$ of $\{ f^*(U_i) \}_{i : I}$, we have that $f^*(\bigvee_i U_i) \le V$ since $f^*(\bigvee_i
  U_i) \le V \leftrightarrow \left(\bigvee_i U_i\right) \le f_*(V)$ meaning it is sufficient to show
  $U_i \le f_*(V)$ for each $U_i$. Since $U_i \le f_*(V)$ iff $f^*(U_i) \le V$, we are
  done as the latter can be seen to hold directly from the fact that $V$ is an
  upper bound of $\{ f^*(U_i) \}_{i : I}$.

  For the converse, suppose $f^*(\bigvee_i U_i) = \bigvee_{i : I} f^*(U_i)$ for every family
  $\{ U_i \}_{i : I}$. We define the right adjoint of $f^*$ as:
  \begin{equation*}
    f_*(V) \quad\is\quad \bigvee \left\{ B_i \mid i : I, f^*(B_i) \le V \right\}.
  \end{equation*}
  We need to show that $f_*$ is the right adjoint of $f^*$ i.e.\ that $f^*(U) \le
  V \leftrightarrow U \le f_*(V)$ for any two $U, V : \opens{X}$. For the forward direction,
  assume $f^*(U) \le V$. We know that there exists a covering family $\{ B_j \}_{j
    \in J}$ for $U$ with $U = \bigvee_{j \in J} B_j$ so it suffices to show that $B_j \le
  f_*(V)$ for every $j \in J$. It remains to show that $f^*(B_j) \le V$. This
  follows from the fact that $f^*(B_j) \le f^*(\bigvee_{j \in J} B_j) \le f^*(U) \le V$. For
  the backward direction, let $U \le f_*(V)$. We have:
  \begin{align*}
    f^*(U) \quad&\le\quad f^*(f_*(V)) \\
         \quad&\equiv\quad f^*\left(\bigvee \left\{ B_i \mid f^*(B_i) \le V \right\}\right) \\
         \quad&\le\quad \bigvee \setof{f^*(B_i) \mid f^*(B_i) \le V}
              && [\text{since $f^*$ preserves joins}] \\
         \quad&\le\quad V.
  \end{align*}
\end{proof}

Our primary use case for the Adjoint Functor Theorem is the construction of
Heyting implications in locally small frames with small bases.

\begin{defn}[Heyting implication]\label{defn:heyting-implication}
  Let $X$ be a large, locally small, and small complete locale with a small
  basis and let $U : \opens{X}$. As the map $\blank \wedge U : \opens{X} \rightarrow \opens{X}$
  preserves joins by the frame distributivity law, it must have a right adjoint
  $h : \opens{X} \rightarrow \opens{X}$, by Theorem \ref{thm:aft}, that satisfies $W \wedge U \le
  V \leftrightarrow W \le h(V)$ for all $W, V : \opens{X}$. We then define the Heyting
  implication as: $U \Rightarrow V \is h(V)$.
\end{defn}

The Adjoint Functor Theorem also allows us to define the notion of a perfect
frame homomorphism.

\begin{defn}[Perfect frame homomorphism]\label{defn:perfect-map}
  Let $X$ and $Y$ be two large, locally small, and small complete locales and
  assume that $Y$ has a small basis. A continuous map $f : X \rightarrow Y$ is said to be
  \define{perfect} if the right adjoint $f_*$ of its defining frame homomorphism
  $f^*$ is Scott continuous.
\end{defn}

\begin{prop}\label{prop:perfect-resp-way-below}
  Let $f : X \rightarrow Y$ be a perfect map where $Y$ is a locale with small basis. The
  frame homomorphism $f^*$ respects the \emph{way below} relation, that is, $U \WayBelow V$
  implies $f^*(U) \WayBelow f^*(V)$, for any two $U, V : \opens{Y}$.
\end{prop}

A proof of Proposition~\ref{prop:perfect-resp-way-below} can be found in
\cite{patch-short}. Our proof is mostly the same, once it is ensured that the
Heyting implication exists through the small basis assumption. We thus omit the
proof.

\begin{cor}\label{cor:perfect-maps-are-spectral}
  Perfect maps are spectral as they preserve compact opens.
\end{cor}

In fact, the converse is also true in the case of spectral locales so
Corollary~\ref{cor:perfect-maps-are-spectral} can be strengthened to an
equivalence in this case.

\begin{prop}\label{lem:perfect-iff-spectral}
  Let $X$ and $Y$ be two large, locally small, and small complete spectral
  locales and assume that $Y$ has a small basis. A continuous map $f : X \to Y$ is
  perfect iff it is spectral.
\end{prop}
\begin{proof}
  The forward direction is given by
  Corollary~\ref{cor:perfect-maps-are-spectral}. For the backward direction,
  assume $f : X \to Y$ to be a spectral map. We have to show that the right
  adjoint $f_* : \opens{X} \to \opens{Y}$ of its defining frame homomorphism is
  Scott continuous. Let $\{ U_i \}_{i : I}$ be a directed family in $\opens{X}$.
  We have to show $f_*(\bigvee_{i : I} U_i) = \bigvee_{i : I} f_*(U_i)$. The
  $\bigvee_{i : I} f_*(U_i) \le f_*(\bigvee_{i : I} U_i)$ direction is immediate. For the
  $f_*(\bigvee_{i : I} U_i) \le \bigvee_{i : I} f_*(U_i)$ direction, let $C$ be a compact open
  such that $C \le f_*(\bigvee_{i : I} U_i)$. By the fact that $f^* \dashv f_*$, it must be
  the case that $f^*(C) \le \bigvee_{i : I} U_i$ and since $f^*(C)$ is compact, by the
  spectrality assumption of $f^*$, there must exist some $l : I$ such that
  $f^*(C) \le U_l$. Again by adjointness, $C \le f_*(U_l)$ so clearly
  $C \le \bigvee_{i : I}  f_*(U_i)$.
\end{proof}

\section{Meet-semilattice of Scott continuous nuclei}\label{sec:meet-semilattice}

In this section, we take the first step towards constructing the defining frame
of the patch locale on a spectral locale i.e.\ the frame of Scott continuous
nuclei. We construct the meet-semilattice of \emph{all} nuclei on a frame.

\begin{prop}\label{defn:nuclei-semilattice}
  The type of nuclei on a given frame $\opens{X}$ forms a meet-semilattice under
  the pointwise order.
\end{prop}
\begin{proof}
  We need to show that the type $\opens{X}$ has all finite meets. The top
  nucleus is defined as $\blank \mapsto \top$ and the meet of two nuclei as
  $j \wedge k \is U \mapsto j(U) \wedge k(U)$. It is easy to see that $j \wedge k$ is the greatest
  lower bound of $j$ and $k$ so it remains to show that $j \wedge k$ satisfies the
  nucleus laws.

  The inflation property can be seen to be satisfied from the inflation
  properties of $j$ and $k$ combined with the fact that $j(U) \wedge k(U)$ is the
  greatest lower bound of $j(U)$ and $k(U)$. To see that meet preservation
  holds, let $U, V : \opens{X}$; we have:
  \begin{align*}
    (j \wedge k)(U \wedge V) &\quad\equiv\quad j (U \wedge V) \wedge k (U \wedge V)         \\
                   &\quad=\quad j(U) \wedge j(V) \wedge k(U) \wedge k(V)     \\
                   &\quad=\quad (j(U) \wedge k(U)) \wedge (j(V) \wedge k(V)) \\
                   &\quad\equiv\quad (j \wedge k)(U) \wedge (j \wedge k)(V).
  \end{align*}
  For idempotency, let $U : \opens{X}$. We have:
  \begin{align*}
    (j \wedge k)((j \wedge k)(U)) &\quad\equiv\quad j (j(U) \wedge k(U)) \wedge k(j(U) \wedge k(U))      \\
                        &\quad=\quad j(j(U)) \wedge j(k(U)) \wedge k(j(U)) \wedge k(k(U)) \\
                        &\quad\le\quad j(j(U)) \wedge k(k(U))                     \\
                        &\quad=\quad j(U) \wedge k(U)                           \\
                        &\quad\equiv\quad (j \wedge k)(U).
  \end{align*}
\end{proof}



We now show that this meet-semilattice can be \emph{refined} to only those
nuclei that are Scott continuous (i.e.\ the \emph{perfect} nuclei).

\begin{prop}\label{prop:sc-nuclei-semilattice}
  The Scott continuous nuclei on any locale form a meet-semilattice.
\end{prop}
\begin{proof}
  Let $X$ be a locale. The construction is the same as the one from
  Proposition~\ref{defn:nuclei-semilattice}; the top element is $\blank \mapsto \top$
  which is trivially Scott continuous so it remains to show that the meet of two
  Scott continuous nuclei is Scott continuous. Consider two Scott continuous
  nuclei $j$ and $k$ on $\opens{X}$ and a directed small family $\{ U_i \}_{i :
    I}$. We then have:
  \begin{align*}
       (j \wedge k) \paren{\bigvee_{i : I} U_i}
  &\quad\equiv\quad j \paren{\bigvee_{i : I} U_i} \wedge k \paren{\bigvee_{j : I} U_j}
     && \\
  &\quad=\quad \paren{\bigvee_{i : I} j(U_i)} \wedge \paren{\bigvee_{j : I} k(U_j)}
     && [\text{Scott continuity of $j$ and $k$}]\\
  &\quad=\quad \bigvee_{(i, j) : I \times I} j(U_i) \wedge k(U_j)
     && [\text{distributivity}]\\
  &\quad=\quad \bigvee_{i : I} j(U_i) \wedge k(U_i)
     && [\text{\dag}] \\
  &\quad\equiv\quad \bigvee_{i : I} (j \wedge k)(U_i)
     && [\text{meet preservation}].
  \end{align*}
  where, for the \dag\ step, we use antisymmetry. The backwards direction is
  immediate. For the forwards direction, we need to show that $\bigvee_{(i, j) : I \times
    I} j(U_i) \wedge k(U_j) \le \bigvee_{i : I} j(U_i) \wedge k(U_i)$, for which it suffices to
  show that $\bigvee_{i : I} j(U_i) \wedge k(U_i)$ is an upper bound of $\{j(U_i) \wedge
    k(U_j)\}_{(i, j) : I \times I}$. Let $m, n : I$ be two indices. As $\{U_i\}_{i :
    I}$ is directed, there must exist some $o$ such that $U_o$ is an upper bound
  of $\{U_m, U_n\}$. Using the monotonicity of $j$ and $k$, we get $j(U_m) \wedge
  k(U_n) \le j(U_o) \wedge k(U_o) \le \bigvee_{i : I} j(U_i) \wedge k(U_i)$ as desired.
\end{proof}

\section{Joins in the frame of Scott continuous nuclei}\label{sec:joins}

The nontrivial component of the patch frame construction is the join of a family
$\{ k_i \}_{i : I}$ of perfect nuclei, as the pointwise join fails to be
idempotent in general, and not inflationary when the family in consideration is
empty.

A construction of the join, given in~\cite{properly-injective}, is
based on the idea that, if we start with a family $\{k_i\}_{i:I}$ of
nuclei, we can consider the family
\begin{equation*}
  \left\{ k_{i_0} \circ \cdots \circ k_{i_n} \right\}_{(i_0, \cdots, i_n) : \mathsf{List}(I)},
\end{equation*}
whose index type is the type of lists of indices in $I$, that will \emph{always}
be directed. We will use the following notation for lists over a type $X$:
\begin{itemize}
  \item $\emptyl$ denotes the empty list,
  \item $x \cons s$, with $x : X$ and $s : \ListTy{X}$, denotes the list with
    first element $x$ followed by the elements of $s$,
  \item $s \append t$ denotes the concatentation of lists $s$ and $t$.
\end{itemize}

To define the join operation, we will use the iterated composition function
$\ddnm$ that we define as follows:
\begin{defn}[Iterated composition of nuclei]
  Given a small family $K \is \{ k_i \}_{i : I}$ of nuclei on a given
  locale $X$, we denote by \define{$K^*$} the family $(\ListTy{I}, \ddnm)$
  where $\ddnm$ is defined as follows:
  \begin{align*}
    \dd{\emptyl}   \quad&\is\quad \mathsf{id};  \\
    \dd{i \cons s} \quad&\is\quad \dd{s} \circ k_i.
  \end{align*}
\end{defn}
By an easy proof by induction, we have the following.
\begin{prop}\label{prop:app-lemma}
  For any family $K \is \{ k_i \}_{i : I}$ of prenuclei on a locale and
  any $s, t : \mathsf{List}(I)$, we have that \(\dd{s \append t} = \dd{t} \circ \dd{s}.\)
\end{prop}

\begin{prop}\label{prop:star-prenucleus}
  Given a family $K \is \{ k_i \}_{i : I}$ of nuclei on a locale,
  every $\alpha \in K^*$ is a prenucleus, that is, for every $s : \mathsf{List}(I)$,
  the function $\dd{s}$ is a prenucleus.
\end{prop}
\begin{proof}
  If $s = \emptyl$, we are done as it is immediate that the identity function
  $\mathsf{id}$ is a prenucleus. If $s = i \cons s'$, we need to show that
  $\dd{s'} \circ k_i$ is a prenucleus. For meet preservation, let $U, V
  : \opens{X}$. We have that:
  \begin{align*}
    (\dd{s'} \circ k_i)(U \wedge V)
  &\quad\equiv\quad \dd{s'}(k_i(U \wedge V))                     \\
  &\quad=\quad \dd{s'}(k_i(U) \wedge k_i(V))
      && [\text{$k_i$ is a nucleus}]          \\
  &\quad=\quad \dd{s'}(k_i(U)) \wedge \dd{s'}(k_i(V))
      && [\text{inductive hypothesis}]        \\
  &\quad\equiv\quad (\dd{s'} \circ k_i)(U) \wedge (\dd{s'} \circ k_i)(V).
  \end{align*}
  For the inflation property, consider some $U : \opens{X}$. We have that $U \le
  k_i(U) \le \dd{s'}(k_i(U))$, by the inflation property of $k_i$ and the
  inductive hypothesis.
\end{proof}

\begin{prop}\label{prop:star-ub}
  Given a nucleus $j$ and a family $K \is \{ k_i \}_{i : I}$ of nuclei on a
  locale, if $j$ is an upper bound of $K$ then it is also an upper bound of
  $K^*$.
\end{prop}
\begin{proof}
  Let $j$ and $K \is \{ k_i \}_{i : I}$ be, respectively, a nucleus and a family
  of nuclei on a locale. Let $s : \mathsf{List}(I)$. We denote by $\{ \alpha_s \}_{s
    : \ListTy{S}}$ the family $K^*$. We proceed by induction on $s$. If $s =
  \emptyl$, we have that $\mathsf{id}(U) \equiv U \le j(U)$. If $s = i \cons s'$, we
  then have:
  \begin{align*}
    \alpha_{s'}(k_i(U))
  \quad&\le\quad \alpha_{s'}(j(U))
    && [\text{monotonicity of $\alpha_{s'}$ (Prop.~\ref{prop:star-prenucleus} and monotonicity of prenuclei)}] \\
  \quad&\le\quad j(j(U))
    && [\text{inductive hypothesis}] \\
  \quad&\le\quad j(U)
    && [\text{idempotency of $j$}].
  \end{align*}
\end{proof}

\begin{prop}\label{prop:star-sc}
  Given a family $\{ k_i \}_{i : I}$ of Scott continuous nuclei on a locale,
  every prenucleus $\alpha \in K^*$ is Scott continuous.
\end{prop}
\begin{proof}
  Any composition of finitely many Scott continuous functions is Scott continuous.
\end{proof}

\begin{prop}\label{prop:star-dir}
  Given a family $K :\equiv \{ k_i \}_{i : I}$ of nuclei on a locale,
  the family $K^*$ is directed.
\end{prop}
\begin{proof}
  $K^*$ is indeed always inhabited by the identity nucleus. The upper bound of
  nuclei $\dd{s}$ and $\dd{t}$ is given by $\dd{s \append t}$, which is $\dd{t}
  \circ \dd{s}$ by Proposition~\ref{prop:app-lemma}. The fact that this is an upper
  bound of $\{ \dd{s}, \dd{t} \}$ follows from monotonicity and inflationarity.
\end{proof}

\begin{prop}\label{prop:delta-gamma}
  Let $j$ be a nucleus and $K \is \{ k_i \}_{i : I}$ a family of nuclei on
  a locale. Denote by $\{ \alpha_s \}_{s : \ListTy{I}}$ the family $K^*$ and by
  $\{ \beta_s \}_{s : \ListTy{I}}$ the family $\{ j \wedge k \mid k \in K \}^*$. We have that
  $\beta_{s}$ is a lower bound of $\{ \alpha_{s}, j \}$ for every $s : \ListTy{I}$.
\end{prop}

We are now ready to construct the join operation in the meet-semilattice of
Scott continuous nuclei hence defining the patch frame $\opens{\Patch(X)}$ of
the frame of a locale $X$.

\begin{thm}[Join of Scott continuous nuclei]\label{defn:sc-join}
  Let $K \is \{ k_i \}_{i : I}$ be a family of Scott continuous nuclei.
  The join of $K$ can be calculated as \(\bigvee^N K \is U \mapsto \bigvee_{\alpha \in K^*} \alpha(U)\).
\end{thm}
\begin{proof}
  It must be checked that this is (1) indeed the join, (2) is a Scott continuous
  nucleus i.e.\ it is inflationary, binary-meet-preserving, idempotent, and
  Scott continuous. The inflation property is direct. For meet preservation,
  consider some $U, V : \opens{X}$. We have:
  \begin{align*}
        \left(\bigvee^N_{i : I} k_i\right)(U \wedge V)
    &\quad\equiv\quad \bigvee_{\alpha \in K^*} \alpha(U \wedge V)   \\
    &\quad=\quad \bigvee_{\alpha \in K^*} \alpha(U) \wedge \alpha(V)
        && [\text{Proposition~\ref{prop:star-sc}}]  \\
    &\quad=\quad \bigvee_{\beta, \gamma \in K^*} \beta(U) \wedge \gamma(V)
        && [\dag]                           \\
    &\quad=\quad \paren{\bigvee_{\beta \in K^*} \beta(U)} \wedge \paren{\bigvee_{\gamma \in K^*} \gamma(V)}
        && [\text{distributivity}]       \\
    &\quad\equiv\quad \paren{\bigvee^N_{i : I} k_i}(U) \wedge \paren{\bigvee^N_{i : I} k_i}(V),
  \end{align*}
  where the step ($\dag$) uses antisymmetry. The
    \(\bigvee_{\alpha \in K^*} \alpha(U) \wedge \alpha(V) \le \bigvee_{\beta, \gamma \in K^*} \beta(U) \wedge \gamma(V)\)
  direction is direct whereas for the
    \(\bigvee_{\beta, \gamma \in K^*} \beta(U) \wedge \gamma(V) \le \bigvee_{\alpha \in K^*} \alpha(U) \wedge \alpha(V)\)
  direction we show that $\bigvee_{\alpha \in K^*} \alpha(U) \wedge \alpha(V)$ is an upper bound of the set
  $\{ \beta(U) \wedge \gamma(V) \mid \beta, \gamma \in K^* \}$. Consider arbitrary $\beta, \gamma \in K^*$. By the
  directedness of $K^*$ we know that there exists some $\delta \in K^*$ that is an
  upper bound of $\{\beta, \gamma\}$. We then have:
  \(\beta(U) \wedge \gamma(V) \le \delta(U) \wedge \delta(V) \le \bigvee_{\alpha \in K^*} \alpha(U) \wedge \alpha(V)\).
  For idempotency, let $U : \opens{X}$. We have that:
  \begin{align*}
         \left(\bigvee^N_{i} k_i\right)\left(\left(\bigvee^N_{i} k_i\right)(U)\right)
    &\quad\equiv\quad \bigvee_{\alpha \in K^*} \alpha \left( \bigvee_{\beta \in K^*} \beta(U) \right) \\
    &\quad=\quad \bigvee_{\alpha \in K^*} \bigvee_{\beta \in K^*} \alpha(\beta(U)) & [\text{Proposition~\ref{prop:star-sc}}]\\
    &\quad\le\quad \bigvee_{\alpha, \beta \in K^*} \alpha(\beta(U)) & [\text{flattening joins}]\\
    &\quad\le\quad \bigvee_{\alpha \in K^*} \alpha(U) & [\dag]\\
    &\quad\equiv\quad \paren{\bigvee^N_i k_i}(U),
  \end{align*}
  where for the step ($\dag$) it suffices to show that $\bigvee_{\alpha \in K^*} \alpha(U)$ is an upper
  bound of the family $\setof{ \alpha(\beta(U)) \mid (\alpha, \beta) \in K^* \times K^* }$. Consider
  arbitrary $\alpha, \beta \in K^*$. There must be lists $s$ and $t$ of indices of $K$ such
  that $\alpha \equiv \dd{s}$ and $\beta \equiv \dd{t}$. We pick $\delta \is \dd{t \append s} \in K^*$ which
  is then an upper bound of $\dd{s}$ and $\dd{t}$ (as in
  Proposition~\ref{prop:star-dir}). By Proposition~\ref{prop:app-lemma}, we have
  that $\dd{t}(\dd{s}(U)) \equiv \dd{t \append s}(U) \equiv \delta(U) \le \bigvee_{\alpha \in K^*}\alpha(U)$.

  For Scott continuity, let $\{ U_j \}_{j : J}$ be a directed family over
  $\opens{X}$. The result then follows as:
  \begin{align*}
    \left(\bigvee^N K\right)\left(\bigvee_{j : J} U_j\right)
  &\quad\equiv\quad \bigvee_{\alpha \in K^*} \alpha\left(\bigvee_{j : J} U_j\right) \\
  &\quad=\quad \bigvee_{\alpha \in K^*} \bigvee_{j : J} \alpha(U_j) && [\text{Proposition~\ref{prop:star-sc}}]    \\
  &\quad=\quad \bigvee_{j : J} \bigvee_{\alpha \in K^*} \alpha(U_j) && [\text{joins commute}] \\
  &\quad\equiv\quad \bigvee_{j : J} \left(\bigvee^N K\right)(U_j)
  \end{align*}
  as required.

  The fact that $\bigvee^N_i k_i$ is an upper bound of $K$ is easy to verify: given
  some $k_i$ and $U : \opens{X}$, $k_i(U) \in \{ \alpha(U) \mid \alpha \in K^* \}$ since $k_i \in
  K^*$. To see that it is \emph{the least} upper bound, consider a
  Scott continuous nucleus $j$ that is an upper bound of $K$. Let $U :
  \opens{X}$. We need to show that $(\bigvee^N_i k_i)(U) \le j(U)$. We know by
  Proposition~\ref{prop:star-ub} that $j$ is an upper bound of $K^*$, since it
  is an upper bound of $K$, which is to say $K^*_{s}(U) \le j(U)$ for every $s :
  \ListTy{I}$ i.e.\ $j(U)$ is an upper bound of the family $\setof{\alpha(U) \mid \alpha \in
    K^*}$. Since $(\bigvee^N_i k_i)(U)$ is the least upper bound of this family by
  definition, it follows that it is below $j(U)$.
\end{proof}

We use Proposition~\ref{prop:delta-gamma} to prove the following.

\begin{prop}[Distributivity]\label{prop:distributivity}
  For any Scott continuous nucleus $j$ and any family $\{ k_i \}_{i : I}$ of Scott continuous nuclei, we have that \[ j \wedge \left(\bigvee_{i : I} k_i\right) = \bigvee_{i : I} j \wedge
  k_i.\]
\end{prop}

It follows
that the Scott continuous nuclei form a frame.

\begin{defn}[Patch locale of a spectral locale]\label{defn:patch}
  Let $X$ be a large, locally small, and small complete spectral locale. The
  \emph{patch locale} of $X$, written $\Patch(X)$, is given by the frame of
  Scott continuous nuclei on~$X$.
\end{defn}

Note that we do not assume the locale $X$ in Definition~\ref{defn:patch} to be
spectral. This is to highlight the fact that the construction of the patch frame
does not rely on this assumption in a crucial way. Nevertheless, the patch
locale is meaningful only on spectral locales as its universal property can be
proved only under the assumption of spectrality.

Definition~\ref{defn:patch} gives rise to a problem that we need to address: the
patch of a locally small locale does not yield a locally small locale. Starting
with a $(\UU^+, \UU, \UU)$-locale $X$, $\Patch(X)$ is a $(\UU^+, \UU^+,
\UU)$-locale since the pointwise ordering of nuclei (defined in
Proposition~\ref{defn:nuclei-semilattice}) quantifies over arbitrary opens. In
most of our development, we have restricted attention to only locally small
frames meaning we run into problems if $\Patch(X)$ is not locally small
(e.g.\ applying the Adjoint Functor Theorem). We circumvent this by using the
following small version of the same relation:

\begin{defn}[Basic nuclei ordering on spectral locales]
\label{defn:nuclei-basic-order}
  Let $X$ be a spectral locale and denote its basis by $\{ B_i \}_{i : I}$. Let
  $j, k : \opens{X} \rightarrow \opens{X}$ be two nuclei. We define the \define{basic
  nuclei ordering} $\blank \le_{\mathfrak{K}} \blank$ as
  \begin{equation*}
    j \le_{\mathfrak{K}} k \quad\is\quad \prod_{i : I} j(B_i) \le k(B_i).
  \end{equation*}
\end{defn}

Given two nuclei $j$ and $k$ on a $(\UU, \VV, \WW)$-locale, the relation $j
\le_{\mathfrak{K}} k$ lives in universe $\VV \vee \WW$ meaning, in the case of a $(\UU^+, \UU,
\UU)$-locale, it lives in $\UU$ as desired.

\begin{prop}\label{prop:basic-ordering-iff}
  The basic nuclei ordering given in Definition~\ref{defn:nuclei-basic-order} is
  logically equivalent to the pointwise ordering of nuclei.
\end{prop}
\begin{proof}
  The usual pointwise ordering obviously implies the basic ordering so we
  address the other direction. Let $j$ and $k$ be two Scott continuous nuclei on
  a spectral locale $X$ and assume that $j \le_{\mathfrak{K}} k$. We need to show that $j(U)
  \le k(U)$ for every open $U$ so let $U : \opens{X}$. It must be the case that $U
  = \bigvee_{l \in L} B_l$ where $\{ B_l \}_{l \in L}$ is the directed basic covering
  family of compact opens covering $U$. We then have $j(\bigvee_{l \in L} B_l) = \bigvee_{l \in
    L} j(B_l)$ by Scott continuity and $\bigvee_{l \in L} j(B_l) \le \bigvee_{l \in L} k(B_l)$
  since $j(B_l) \le k(B_l)$ for every $l : L$.
\end{proof}

Thanks to Proposition~\ref{prop:basic-ordering-iff} our theorems that have the
local smallness assumption apply to the patch frame as we know that $\Patch(X)$
always has an equivalent copy that is locally small. We also note that we
will not always be precise in distinguishing between the basic order and the regular
order on nuclei and will freely switch between the two, making implicit use of
Proposition~\ref{prop:basic-ordering-iff}.

\section{The coreflection property of \texorpdfstring{$\Patch$}{Patch}}
\label{sec:coreflection}

We prove in this section that our construction of $\Patch$ has the desired
universal property: it exhibits $\Stone$ as a coreflective subcategory of
$\Spec$. We also note that when we talk about Stone and spectral locales in this
section, we implicitly assume them to be large, locally small, and small
complete, and refrain from explicitly stating this assumption.

The notions of \emph{closed} and \emph{open} nuclei are crucial for proving the
universal property. We first give the definitions of these. Let $U$ be an open
of a locale $X$;
\begin{enumerate}
  \item The \emph{closed nucleus} induced by $U$ is the map $V \mapsto U \vee V$;
  \item The \emph{open nucleus} induced by $U$ is the map $V \mapsto U \Rightarrow V$.
\end{enumerate}
We denote the closed nucleus induced by $U$ by $\closednucl{U}$ and, because the
open nucleus is the Boolean complement of the closed nucleus, we denote the open
nucleus induced by $U$ by $\opennucl{U}$. This follows the notation of
\cite{patch-short, patch-full}. We now prove the Scott continuity of these
nuclei.
\begin{lem}
\label{lem:characterisation}
  For any spectral locale $X$ and any monotone map $h : \opens{X} \rightarrow \opens{X}$,
  if for every $U : \opens{X}$ and compact $C : \opens{X}$ with $C \le h(U)$,
  there is some compact $D \le U$ such that $C \le h(D)$, then $h$ is Scott
  continuous
\end{lem}

\begin{lem}
  Let $X$ be a spectral locale. The closed nucleus $\closednucl{U}$ on $X$ is
  Scott continuous for any open $U$, whereas the open nucleus is
  Scott continuous if the open $U$ is compact.
\end{lem}
\begin{proof}~

  \emph{Closed nucleus}. Let $U$ be an open of a locale and let $\{ V_i \}_{i :
    I}$ be a directed family of opens. We need to show that $\closednucl{U}(\bigvee_{i
    : I} V_i) = \bigvee_{i : I} \closednucl{U}(V_i)$. It is clear that $U \vee (\bigvee_{i : I}
  V_i)$ is an upper bound of $\{ U \vee V_i \}_{i : I}$. Let $W$ be an arbitrary
  upper bound of $\{ U \vee V_i \}_{i : I}$. It suffices to show that $W$ is an
  upper bound of $\{U, (\bigvee_{i : I} V_i)\}$. For the case of $\bigvee_{i : I} V_i$, we
  have that $\bigvee_{i : I} V_i \le \bigvee_{i : I} U \vee V_i \le W$. For the case of $U$, we use
  the fact that $\{ V_i \}_{i : I}$ is directed. Since the family $\{ V_i \}_{i
    : I}$ is directed it must be inhabited by some $V_k$. We then have $U \le U \vee
  V_k \le W$ as $W$ is an upper bound of $\{ U \vee V_i \}_{i : I}$.

  \emph{Open nucleus}. Let $D$ be a compact open of a locale. By
  Lemma~\ref{lem:characterisation}, it is sufficient to show that for any open
  $V$ and any compact open $C_1$ with $C_1 \le D \Rightarrow V$, there exists some compact
  $C_2 \le D$ such that $C_1 \le D \Rightarrow C_2$. Let $V$ and $C_1$ be two opens with
  $C_1$ compact and satisfying $C_1 \le D \Rightarrow V$. Pick $C_2 \is D \wedge C_1$. We know
  that this is compact by spectrality. It remains to check (1) $C_2 \le V$ and
  (2) $C_1 \le D \Rightarrow C_2$, both of which are direct.
\end{proof}

In Lemma~\ref{lem:eps-perfect}, we prove that the map whose inverse image sends an open $U$ to the
closed nucleus $\closednucl{U}$ is perfect. Before Lemma~\ref{lem:eps-perfect},
we record two auxiliary lemmas that are needed in the proof.

\begin{lem}\label{lem:eps-ra-bot}
  Let $X$ be a spectral locale with a small basis. The right adjoint $\varepsilon_* :
  \opens{\Patch(X)} \to \opens{X}$ of $\closednucl{\blank}$ is equal to
  the assignment $j \mapsto j(\bot)$ i.e. $\varepsilon_*(j) = j(\bot)$ for every Scott continuous nucleus
  $j$ on $X$.
\end{lem}

\begin{lem}\label{lem:directed-computed-pointwise}
  Given a directed family $\{ k_i \}_{i : I}$ of Scott continuous nuclei, their
  join is computed pointwise, that is, $\left(\bigvee_{i : I} k_i\right)(U) = \bigvee_{i :
  I} k_i(U)$.
\end{lem}

Proofs of Lemma~\ref{lem:eps-ra-bot} and
Lemma~\ref{lem:directed-computed-pointwise} can be found in \cite{patch-short}.
They are omitted here as they are mostly unchanged in our type-theoretical
setting.

\begin{lem}\label{lem:eps-perfect}
  The function that sends an open $U$ to the closed nucleus $\closednucl{U}$ is
  a perfect frame homomorphism $\opens{X} \to \opens{\Patch(X)}$.
\end{lem}
\begin{proof}
  We have to show that the right adjoint $\varepsilon_*$ of $\closednucl{\blank}$ is
  Scott continuous. Let $\{ k_i \}_{i : I}$ be a directed family of
  Scott continuous nuclei. By Lemma~\ref{lem:eps-ra-bot}, it suffices to show
  $\left(\bigvee_{i : I} k_i\right)(\bot) = \bigvee_{i : I} \varepsilon_*(k_i)$. By
  Lemma~\ref{lem:directed-computed-pointwise}, we have
  that $\left(\bigvee_{i : I} k_i\right)(\bot) = \bigvee_{i : I} k_i(\bot)$. The desired result
  of $\bigvee_{i : I} k_i(\bot) = \bigvee_{i : I} \varepsilon_*(k_i)$ is then immediate by
  Lemma~\ref{lem:eps-ra-bot}.
\end{proof}

This function defines a continuous map $\varepsilon : \Patch(X) \to X$, which we we will
show to be the counit of the coreflection in consideration.

\subsection{\texorpdfstring{$\Patch$}{Patch} is Stone}

Before we proceed to showing that the $\Patch$ locale has the desired universal
property, we first need to show that $\Patch(X)$ is Stone (as given in
Definition~\ref{defn:stone}) for any spectral locale $X$. We start by addressing
the question of zero-dimensionality.

To show that $\Patch(X)$ is zero-dimensional, we need to construct a basis
consisting of clopens. We will use the following fact, which was already
mentioned above:

\begin{prop}\label{prop:complementation}
  The open nucleus $\opennucl{U}$ is the Boolean complement of the closed
  nucleus $\closednucl{U}$.
\end{prop}

\begin{lem}\label{lem:patch-zero-dimensional}
  The patch of any spectral locale $X$ with a basis
  $\{ B_i \}_{i : I}$ of compact opens is zero-dimensional, with a
  basis of clopens of the form
  $\bigvee_{(m, n) \in M \times N} \closednucl{B_m} \wedge
  \opennucl{B_n}$ with $M$ and $N$ finite, which is clearly closed
  under finite joins.
\end{lem}
More precisely, if the given basis of $X$ is the family $B : I \to \opens{X}$,
then the constructed basis of $\Patch(X)$ is the family $C : \ListTy{I \times I} \to
\opens{\Patch(X)}$ defined by
  \[
    C([(m_0,n_0),\dots,(m_{k-1},n_{k-1})]) \is \bigvee_{0 \le i < k} \closednucl{B_{m_i}} \wedge \opennucl{B_{n_i}}.
  \]
That is, the index set of the basis consists of formal expressions for finite joins.
\begin{proof}
  We need to show that this (1) consists of clopens, and (2) indeed forms a
  basis. For (1), $\closednucl{B_1} \wedge \opennucl{B_2}$ has complement
  $\opennucl{B_1} \vee \closednucl{B_2}$, by
  Proposition~\ref{prop:complementation}, and finite unions of complemented sets
  are complemented. For (2), let $j : \opens{X} \rightarrow \opens{X}$ be a perfect
  nucleus on $\opens{X}$. We need to show that there exists a subfamily of $C$
  that yields $j$ as its join. For this we pick the subfamily $
  \left\{ \closednucl{B_m} \wedge \opennucl{B_n} \mid m, n : I, B_m \le j(B_n) \right\}$.
  The fact that $j$ is the least upper bound of this subfamily follows from
  Lemma~\ref{lem:johnstones-lemma} and Lemma~\ref{lem:last-step}:
  \begin{align*}
    j \quad&=\quad \bigvee_{n : I} \closednucl{j(B_n)} \wedge \opennucl{B_n}
           && [\text{Lemma~\ref{lem:johnstones-lemma}}]\\
      \quad&=\quad \bigvee \left\{
               \closednucl{B_m} \wedge \opennucl{B_n} \mid m, n : I, B_m \le j(B_n)
             \right\}
           && [\text{Lemma~\ref{lem:last-step}}]
  \end{align*}
\end{proof}

The following is adapted from Johnstone~\cite[Proposition II.2.7]{stone-spaces}.
\begin{lem}\label{lem:johnstones-lemma}
  Given any perfect nucleus $j : \Patch(X)$, we have that $j = \bigvee \left\{
  \closednucl{j(B_n)} \wedge \opennucl{B_n} \mid n : I \right\}$.
\end{lem}

\begin{lem}\label{lem:last-step}
  Let $X$ be a spectral locale. Given any perfect nucleus $j : \Patch(X)$, we
  have that
      $\bigvee \left\{ \closednucl{j(B_n)} \wedge \opennucl{B_n} \mid n : I \right\}
    = \bigvee \left\{
          \closednucl{B_m} \wedge \opennucl{B_n} \mid m, n : I, B_m \le j(B_n)
        \right\}$.
\end{lem}

\begin{thm}\label{thm:patch-is-stone}
  Given any spectral locale $X$, we have that $\Patch(X)$ is a Stone locale.
\end{thm}
\begin{proof}
  Zero-dimensionality is given by Lemma~\ref{lem:patch-zero-dimensional} so it
  only remains to show compactness. Recall that the top element
  $\top$ of $\Patch(X)$ is defined as $\top \is \blank \mapsto \top_X$. Because $\varepsilon^*$ is a frame
  homomorphism, it must be the case that $\top = \varepsilon^*(\top_X)$ meaning what we want to
  show is $\varepsilon^*(\top_X) \WayBelow \varepsilon^*(\top_X)$. By
  Proposition~\ref{prop:perfect-resp-way-below}, it suffices to show $\top_X
  \WayBelow \top_X$ which is immediate as spectral locales are compact.
\end{proof}

\subsection{The universal property of the patch construction}\label{subsec:universal}

We now prove the universal property of $\Patch$ corresponding to the fact that
it is the right adjoint to the inclusion $\Stone \hookrightarrow \Spec$. For this purpose, we
use the following lemma, which is not needed in~\cite{patch-short, patch-full}
thanks to the existence of the frame of all nuclei in the impredicative setting.
\begin{lem}\label{lem:extension}
  Let $L, L'$ be two spectral frames and $B$ a small Boolean algebra embedded in
  $L$ such that
  \begin{enumerate}
  \item $L$ is generated by $A$, and
  \item $B$ contains all compact opens of $L$.
  \end{enumerate}
  Then for any lattice homomorphism $h : B \to L'$, there is a unique frame
  homomorphism $\bar{h} : L \to L'$ satisfying $h = \bar{h} \circ \eta$, where $\eta : B \hookrightarrow
  L$ denotes the embedding of $B$ into $L$, as illustrated in the following
  diagram:
  \begin{equation}\label{eqn:diag1}
    \begin{tikzcd}
      B \arrow[dr, swap, "h"]
      \arrow[r, hook, "\eta"] & L \arrow[d, dashed, "\bar{h}"] \\
      & L'
    \end{tikzcd} \tag{\dag}
  \end{equation}

\end{lem}
\begin{proof}
  Define $\bar{h}(x) \is \bigvee \setof{ h(b) \mid \eta(b) \le x, b : B }$. We need to show
  that (1) $\bar{h}$ is a frame homomorphism, and (2) is the unique map
  satisfying $h = \bar{h} \circ \eta$.

  \textbf{(1)} It is clear that $\bar{h}$ preserves $\bot$, $\top$, and joins of
  directed families. To show that it preserves binary joins, we make use of the
  fact that for any $b \le x \vee y$ with $b$ compact (in any spectral locale), there
  exist compact opens $c \le x$ and $d \le y$ such that $b \le c \vee d$. As it preserves
  both binary joins and directed joins, it must preserve arbitrary joins.

  \textbf{(2)} It is easy to see that $\bar{h}$ satisfies the equation $h =
  \bar{h} \circ \eta$. Uniqueness follows from the fact that $\eta$ is injective.
\end{proof}

We can now present the universal property.

\begin{thm}\label{thm:main}
  Given any spectral map $f : X \rightarrow A$ from a Stone locale into a spectral locale,
  there exists a unique spectral map $\bar{f} : X \rightarrow \Patch(A)$ satisfying $\varepsilon \circ
  \bar{f} = f$, as illustrated in the following diagram in the category of
  spectral locales:
  \begin{center}
    \begin{tikzcd}
      X \arrow[d, swap, "f"] \arrow[dr, dashed, "\bar{f}"] & \\
      A & \Patch(A) \arrow[l, "\varepsilon"]
    \end{tikzcd}
  \end{center}
\end{thm}
\begin{proof}
  We apply Lemma~\ref{lem:extension} with $L \is \opens{\Patch(A)}$, $L' \is
  \opens{X}$, $B \is \KK(\Patch(A))$ and $h$ defined by
  \begin{equation*}
    h\left(\bigvee_{(j, k) \in J \times K} \closednucl{B_j} \wedge \opennucl{B_k}\right)
      \quad\is\quad \bigvee_{(j, k) \in J \times K} f^*(B_j) \wedge \neg f^*(B_k).
    \end{equation*}
    It is easy to see that $h$ is well-defined, in the sense that if
    the same clopen is expressed in two different ways as a finite
    join of binary meets, then $h$ gives the same value for them. It
    is easy to check that the embedding
    $\KK(\Patch(A)) \hookrightarrow \opens{\Patch(A)}$ satisfies the
    premise of the lemma. We then take $\bar{f}^*$ to be $\bar{h}$ as
    constructed in the lemma.  We need to show that this satisfies
    $\bar{f}^*(\closednucl{U}) = f^*(U)$ for all $U : \opens{A}$. It
    suffices to consider the case where $U$ is a compact open~$C$, as
    the compact opens form a basis. Because $C$ can be written as
    $\bigvee \{ \closednucl{C} \wedge \opennucl{\bot}\}$, we have that
  \begin{equation*}
    \bar{f}^*(\closednucl{C})
  = h\left(\bigvee \{ \closednucl{C} \wedge \opennucl{\bot} \}\right)
  = \bigvee \{ f^*(C) \wedge \neg f^*(\bot) \}
  = \bigvee \{ f^*(C) \wedge \top \}
  = f^*(C),
  \end{equation*}
  as required.
\end{proof}

\section{Summary and discussion} \label{conclusion}

We have constructed the patch locale of a spectral locale in the predicative and
constructive setting of univalent type theory, using only propositional and
functional extensionality and the existence of quotients. Furthermore, we have
shown that the patch construction $\Patch : \Spec \rightarrow \Stone$ is the right adjoint
to the inclusion $\Stone \hookrightarrow \Spec$ which is to say that patch exhibits the
category $\Stone$ as a coreflective subcategory of $\Spec$.

As we have elaborated in Section~\ref{sec:spec-and-stone}, answering this
question in a predicative setting has involved the reformulation of several
fundamental concepts of locale theory. In particular, we have reformulated
notions of spectrality, zero-dimensionality, and regularity, and have shown that
crucial facts about these notions remain valid in the predicative setting.

We have also formalised almost all of our development, most
importantly Theorem~\ref{thm:patch-is-stone} and
Lemma~\ref{lem:extension}. The formalisation has been carried out by
the first-named author as part\footnote{The HTML rendering of the
  \textsf{Agda} code can be browsed at
  \url{https://www.cs.bham.ac.uk/~mhe/TypeTopology/Locales.index.html}}
of the second-named author's \texttt{TypeTopology}
library~\cite{typetopology}. Almost all of the presented results have
already been implemented, including:
\begin{enumerate}
  \item All of Section~\ref{sec:spec-and-stone} in the module
    \ttmodule{Locales.CompactRegular};
  \item The Adjoint Functor Theorem and its application to define Heyting
    implications in frames (Section~\ref{sec:aft}) in modules
    \ttmodule{Locales.GaloisConnection},
    \ttmodule{Locales.AdjointFunctorTheoremForFrames}, and
    \ttmodule{Locales.HeytingImplication};
  \item All of Section~\ref{sec:meet-semilattice} and Section~\ref{sec:joins}
    in module \ttmodule{Locales.PatchLocale}; and
  \item The extension lemma (Lemma~\ref{lem:extension}) from
    Section~\ref{subsec:universal} in \ttmodule{Locales.BooleanAlgebra}.
\end{enumerate}
The only result that remains to be formalised is the universal property which we
have proved using Lemma~\ref{lem:extension}. The formalisation of this
result is work in progress and is soon to be completed.

In previous work~\cite{patch-short, patch-full}, that forms the basis of the
present work, the patch construction was used to
\begin{enumerate}
  \item exhibit $\Stone$ as a coreflective subcategory of $\Spec$, which we have
    addressed here,
    \label{item:coref-1} and
  \item exhibit the category of compact regular locales and continuous maps as a
    coreflective subcategory of of stably compact locales and perfect maps,
    which we leave for future work. \label{item:coref-2}
\end{enumerate}
Coquand and Zhang~\cite{coquand-zhang} tackled (\ref{item:coref-2}) using formal
topology. We conjecture that it should be possible to instead use the approach
we have presented here, namely, working with locales with small bases and
constructing the patch as the frame of Scott continuous nuclei.

\bibliographystyle{entics}
\bibliography{references}

\end{document}